\documentclass[pdflatex,sn-mathphys-num]{sn-jnl}
\usepackage{graphicx}%
\usepackage{multirow}%
\usepackage{amsmath,amssymb,amsfonts}%
\usepackage{amsthm}%
\usepackage{mathrsfs}%
\usepackage[title]{appendix}%
\usepackage{xcolor}%
\usepackage{textcomp}%
\usepackage{manyfoot}%
\usepackage{booktabs}%
\usepackage{algorithm}%
\usepackage{algorithmicx}%
\usepackage{algpseudocode}%
\usepackage{listings}%
\usepackage{bbm}
\usepackage{bm,stmaryrd,color}
\usepackage{hyperref}
\usepackage{times}
\usepackage{natbib,dsfont}
\usepackage{epstopdf}
\usepackage{graphics}
\usepackage{subfig}
\usepackage{float}
\usepackage{graphicx}
\theoremstyle{thmstyleone}
\newtheorem{theorem}{Theorem}
\newtheorem{proposition}{Proposition}
\newtheorem{lemma}{Lemma}
\theoremstyle{thmstyletwo}

\theoremstyle{thmstylethree}
\newtheorem{corollary}{Corollary}
\theoremstyle{remark}
\newcommand{\E}{E}

\newcommand{\I}{\mathcal{I}}
\newcommand{\T}{T}

\DeclareMathOperator{\expit}{expit}

\DeclareMathOperator{\diag}{diag}
\DeclareMathOperator{\err}{err}
\DeclareMathOperator{\ARE}{ARE}
\DeclareMathOperator{\RE}{RE}

\newcommand{\en}{e}
\DeclareMathOperator*{\argmax}{arg\,max}

\theoremstyle{thmstyleone}

\theoremstyle{thmstyletwo}

\theoremstyle{thmstylethree}

\begin{document}

\title[Article Title]{Favourable Missingness in Semi-Supervised Classification for Exponential Mixture Models}

\author[1]{\fnm{Huanchao} \sur{Zhou}}\email{zhouhc@xtu.edu.cn}

\author[2]{\fnm{Jinran} \sur{Wu}}\email{jinran.wu@uq.edu.au}

\author[3]{\fnm{Fariborz} \sur{Setoudehtazangi}} \email{fariborz.setoudehtazangi@studenti.unipd.it}

\author[2]{\fnm{Geoffrey J.} \sur{McLachlan}} \email{g.mclachlan@uq.edu.au}

\affil[1]{\orgdiv{School of Mathematics and Computational Science}, \orgname{Xiangtan University}, \orgaddress{  \country{China}}}

\affil*[2]{\orgdiv{School of Mathematics and Physics}, \orgname{The University of Queensland}, \orgaddress{ \country{Australia}}}

\affil[3]{\orgdiv{Dipartimento di Scienze Statistiche}, \orgname{Università di Padova}, \orgaddress{\country{Italy}}}

\abstract{
Semi-supervised classifiers are commonly trained from samples in which all features are observed but some class labels are missing. When label missingness is independent of the observed data, unavailable class memberships reduce Fisher information relative to a completely classified sample. We study a different regime in which the probability of label missingness depends on posterior classification uncertainty, so that the observed missing-label indicators can themselves carry information about the Bayes decision boundary. Building on the conditionally weighted information decomposition of Ahfock and McLachlan, we develop this phenomenon for a two-component exponential mixture. Although the exponential model is non-Gaussian, asymmetric, and supported on the positive half-line, its log-posterior odds remain linear in the feature. We derive Bayes' rule and its exact error rate, formulate entropy-logistic and squared-discriminant missingness mechanisms, and obtain the full partially classified likelihood. We then derive a decomposition of the Fisher information into the complete-data information, the conditionally weighted loss due to missing labels, and the information contributed by the missing labels. 
Numerical quadrature identifies regions in which the full likelihood classifier has asymptotic relative efficiency above or below one. Monte Carlo experiments with finite training samples broadly support the population calculations, with the largest departures from the asymptotic predictions occurring near the transition at which the relative efficiency crosses one.
}


\maketitle

\section{Introduction}\label{sec1}
We consider semi-supervised classification in the setting where the features are observed for all units in the training sample, whereas the class labels are available for only a subset of them. Such data arise when routine measurements are inexpensive but definitive labels require expert adjudication, invasive testing, prolonged follow up, or another costly procedure. 
Likelihood based methods for partially classified samples have a long history in discriminant analysis and finite mixture modelling \citep{mclachlan1975iterative, ganesalingam1978efficiency, ONeill1978,McLachlanGordon1989, Dempster1977,mclachlan2019finite}.
Broader accounts of semi-supervised methodology are provided by
\citep{hady2013semi,vanEngelenHoos2020,AhfockMcLachlan2020}.

When labels are missing completely at random (MCAR), the unavailable class memberships must be treated as latent variables, resulting in la loss of Fisher information than that available from a completely classified sample. 
This loss can be substantial when the component distributions overlap, because the unobserved labels then contain considerable information about the discriminant structure \citep{ONeill1978,Louis1982,CastelliCover1995,CastelliCover1996} However, the probability that a label is missing varies across the observed feature space in many applications. In particular, the probability that a label is missing may vary with the observed feature value and may be higher for observations with greater classification uncertainty. Consequently, the observed pattern of missing labels may provide information about the classifier parameters.

Motivated by this observation, \citep{AhfockMcLachlan2020} modeled the probability of a missing label as a function of classification uncertainty. For two normal populations with equal variances, they showed that the observed missing labels can carry information about the separating hyperplane. This additional information can outweigh the loss caused by the unavailable labels in some settings. Consequently, a classifier fitted using the full partially classified likelihood can have a smaller asymptotic expected error rate than a classifier based on the completely classified likelihood that does not use the missing labels. Subsequent work has further developed this framework, principally for Gaussian mixtures, and has considered its broader likelihood implications \citep{Lyu2024,LyuAhfockThompsonMcLachlan2024, WuWangMcLachlan2026}.

Although the underlying likelihood principle is not restricted to normal models, explicit distribution specific analysis have remained concentrated on Gaussian mixtures. Therefore, it is unclear how asymmetry and support restrictions affect the balance between the information lost through missing labels and the information contributed by their missingness pattern. 
Positive and right skewed measurements, including waiting times, lifetimes, and inter arrival times, can often be modeled using exponential mixtures \citep{Jewell1982,mclachlan2019finite}.
Mixtures of distinct exponential components are identifiable up to label permutation under standard conditions \citep{Teicher1961,YakowitzSpragins1968}.
Unlike Gaussian mixtures, exponential mixtures are supported on the positive half line, and the zero of the log posterior odds need not lie in the interior of that support. These features make the exponential mixture a useful non-Gaussian model for studying favourable informative missingness and clarifying the mechanisms through which it can arise in broader classification settings.

In this paper, we consider a two component exponential mixture in which labels are more likely to be missing for observations with greater classification uncertainty. 
Our central question is the extent to which the Fisher information carried by the missing label pattern compensates for the information lost because some class labels are unavailable, after accounting for the nuisance parameters.
Although the component distributions are asymmetric, the log posterior odds remain linear in the observation. 
The Bayes' rule is therefore determined by a single positive threshold when the zero of the log posterior odds lies in the interior of the support, and otherwise reduces to a constant rule. We derive the corresponding Bayes' error and observed data likelihood, and relate the probability of a missing label to the entropy of the posterior class probabilities.
A local expansion of the entropy mechanism also establishes its connection with the squared discriminant model of \citet{AhfockMcLachlan2020}.

Our main result expresses the efficient Fisher information for the discriminant coefficients as the information from complete classification, the conditionally weighted loss caused by unobserved class labels, and the positive semidefinite information carried by the missing labels. 
The conditional weighting is essential because, under feature dependent missingness, observations with missing labels do not form a simple random subsample of the feature distribution. For the exponential mixture, all components of the decomposition admit one dimensional integral representations. 
We then translate the information comparison into the asymptotic relative efficiency of the estimated Bayes' rules. Numerical calculations identify parameter regions in which this relative efficiency is either above or below one, and Monte Carlo experiments examine the corresponding finite sample behaviour.

The remainder of the paper is organised as follows.
Section~\ref{sec:framework} introduces the exponential mixture model, the partially classified likelihoods, the missing label mechanism, and the estimated Bayes' rules.
Section~\ref{sec:information} develops the Fisher information decomposition and its explicit form for the exponential mixture.
Section~\ref{sec:are} derives the asymptotic expected excess error and the resulting relative efficiency, together with its numerical evaluation.
Section~\ref{sec:simulations} reports the finite sample simulation study.
Section~\ref{sec:discussion} concludes with a discussion.

\section{Preliminaries} \label{sec:framework}
\subsection{Exponential mixture model and Bayes' rule}

Let $C_1$ and $C_2$ denote two classes with prior probabilities $\pi_1$ and $\pi_2$ respectively, where 
$$\pi_i=\Pr(Z=i),\qquad \pi_i>0,\qquad \pi_1+\pi_2=1,$$
and $Z$ denotes the class membership. 
Conditional on membership of class $C_i$, the positive feature $Y$ has exponential density, parametrised by its mean,
\begin{equation}\label{eq:expdensity}
 f_i(y;\mu_i)=\frac{1}{\mu_i}\exp\left(-\frac{y}{\mu_i}\right),
 \qquad y>0,\quad \mu_i>0,\quad i=1,2.
\end{equation}
The marginal density of $Y$ is
\begin{equation}
 f_Y(y;\bm\theta)
 =\pi_1f_1(y;\mu_1)+\pi_2f_2(y;\mu_2),
 \qquad
 \bm\theta=(\pi_1,\mu_1,\mu_2)^{\T}.
 \label{eq:mixture-density}
\end{equation}

Let $R(y;\bm\theta)$ denote the Bayes' rule.  
Thus $R(y;\bm\theta)=h$ if
$$ h=\argmax_{i=1,2}\tau_i(y;\bm\theta), $$
where
\begin{equation}
 \tau_i(y;\bm\theta)=\Pr(Z=i\mid Y=y)
 =\frac{\pi_i f_i(y;\mu_i)}{f_Y(y;\bm\theta)}
 \label{eq:posterior-probability}
\end{equation}
is the posterior probability that an entity with feature value $y$ belongs to Class $C_i$.

For the exponential class distributions, the Bayes' rule depends on $\bm\theta$ through the discriminant coefficient vector 
$\bm\beta=(\beta_0,\beta_1)^{\T}$. 
In particular,
\begin{align}
 d(y;\bm\beta)=\log\left\{\frac{\tau_1(y;\bm\theta)}
 {\tau_2(y;\bm\theta)}\right\}=\beta_0+\beta_1y,\label{eq:discriminant}
\end{align}
where
\begin{equation}
 \beta_0=\log\frac{\pi_1}{\pi_2}+\log\frac{\mu_2}{\mu_1},\qquad
 \beta_1=\frac{1}{\mu_2}-\frac{1}{\mu_1}. \label{eq:beta-definition}
\end{equation}
It follows that
\begin{equation}
 \tau_1(y;\bm\beta)
 =\frac{\exp\{d(y;\bm\beta)\}}{1+\exp\{d(y;\bm\beta)\}},
 \qquad
 \tau_2(y;\bm\beta)=\frac{1}{1+\exp\{d(y;\bm\beta)\}}.
 \label{eq:posterior-logistic}
\end{equation}
We henceforth write the reduced form of the Bayes' rule as 
\begin{equation}\label{eq:bayesrule}
 R(y;\bm\beta)= \begin{cases}
                   1,&d(y;\bm\beta)>0,\\
                   2,&d(y;\bm\beta)\leq0.
                \end{cases}
\end{equation}

When $\mu_1=\mu_2$, the two class conditional distributions coincide and the feature contains no information about class membership. Therefore, we focus on the case
\begin{equation}
 \mu_1<\mu_2.\label{eq:ordering}
\end{equation}
Under this ordering, $\beta_1<0$ and $d(y;\bm\beta)$ is strictly decreasing on $(0,\infty)$.

Define the limiting posterior odds
\begin{equation}
	r_0=\frac{\pi_1\mu_2}{\pi_2\mu_1}. \label{eq:r0def}
\end{equation}
Then
$$ \lim_{y\downarrow0} \frac{\tau_1(y;\bm\theta)}{\tau_2(y;\bm\theta)} =r_0, $$
and $\beta_0=\log r_0$. Hence the sign of $r_0-1$ determines whether the Bayes decision boundary lies in the interior of the support. 
If $r_0>1$, the boundary is
\begin{equation}
	t_0 =-\frac{\beta_0}{\beta_1} =\frac{\mu_1\mu_2}{\mu_2-\mu_1}
	\log\left(\frac{\pi_1\mu_2}{\pi_2\mu_1}\right)>0. \label{eq:t0}
\end{equation}

\begin{proposition}[Bayes' rule and error rate] \label{prop:bayeserror}
	Assume that $\mu_1<\mu_2$.
	\begin{enumerate}[label=(\roman*)]
		\item If $r_0>1$, the Bayes' rule assigns $0<y<t_0$ to class
		$C_1$ and $y\geq t_0$ to class $C_2$. 
        Its optimal overall error rate is
		\begin{align}
			\err(\bm\theta) &= \pi_1\Pr\{R(Y;\bm\beta)=2\mid Z=1\}
			+\pi_2\Pr\{R(Y;\bm\beta)=1\mid Z=2\} \notag\\
			&=\pi_1\exp(-t_0/\mu_1)+\pi_2\{1-\exp(-t_0/\mu_2)\}.
			\label{eq:bayeserror}
		\end{align}
		
		\item If $r_0\leq1$, the Bayes' rule assigns every $y>0$ to class $C_2$, and $\err(\bm\theta)=\pi_1$.
	\end{enumerate}
\end{proposition}
\noindent The proof is given in Appendix~\ref{app:proof-bayes}.

\subsection{Partially classified samples and likelihoods}
\label{subsec:likelihoods}
Let $$ x_{\rm CC}=(x_1^{\T},\ldots,x_n^{\T})^{\T}, \quad 
x_j=(y_j,z_j)^{\T}, \quad j = 1, \dots, n $$
denote a completely classified sample of $n$ independent observations from $X=(Y,Z)^{\T}$. 
Write $z_{ij}=1$ if $z_j=i$ and zero otherwise. For a partially classified sample, let $m_j$ be the missing label indicator, equal to one if $z_j$ is missing and zero if it is available. 
The partially classified sample $x_{\rm PC}$ therefore contains $(y_j,z_j)$ when $m_j=0$ and only $y_j$ when $m_j=1$.

Let $M_j$ denote the random variable corresponding to $m_j$. 
Under the assumption
\begin{equation}
 \Pr(M_j=1\mid y_j,z_j) =\Pr(M_j=1\mid y_j) =q(y_j;\bm\beta,\bm\xi),
 \label{eq:MAR}
\end{equation}
where $\bm\xi$ denotes the parameter of the missing label mechanism, so that $M_j\perp Z_j\mid Y_j$. 
The entropy based specification of $q$ is introduced in Section~\ref{sec:mechanism}.

The log likelihood for the completely classified sample is
\begin{equation}
 \log L_{\rm CC}(\bm\theta)=\sum_{j=1}^n\sum_{i=1}^2 z_{ij}
 \log\{\pi_i f_i(y_j;\mu_i)\}. \label{eq:cclik}
\end{equation}
The log likelihood based on the partially classified sample that ignores the missing label mechanism is
\begin{align}
 \log L_{\rm PC}^{(\mathrm{ig})}(\bm\theta) =
 \sum_{j=1}^n(1-m_j)\sum_{i=1}^2 z_{ij}\log\{\pi_i f_i(y_j;\mu_i)\}
 +\sum_{j=1}^n m_j\log f_Y(y_j;\bm\theta).\label{eq:iglik}
\end{align}
The likelihood formed from the missing label indicators has logarithm
\begin{equation}
 \log L_{\rm PC}^{(\mathrm{miss})}(\bm\beta,\bm\xi)
 =\sum_{j=1}^n \left[(1-m_j)\log\{1-q(y_j;\bm\beta,\bm\xi)\}
 +m_j\log q(y_j;\bm\beta,\bm\xi) \right]. 
 \label{eq:missing-loglikelihood}
\end{equation}
Equivalently, the observed data likelihood contribution of the $j$th observation is
\begin{align}
 \left[\prod_{i=1}^2\{\pi_i f_i(y_j;\mu_i)\}^{z_{ij}}
 \{1-q(y_j;\bm\beta,\bm\xi)\}\right]^{1-m_j}
 \left[ f_Y(y_j;\bm\theta)q(y_j;\bm\beta,\bm\xi)\right]^{m_j}.
 \label{eq:observedcontribution}
\end{align}
Thus, with $\bm\Psi=(\bm\theta^{\T},\bm\xi^{\T})^{\T}$, the full log likelihood is
\begin{equation}
 \log L_{\rm PC}^{(\mathrm{full})}(\bm\Psi) 
 =\log L_{\rm PC}^{(\mathrm{ig})}(\bm\theta)
 +\log L_{\rm PC}^{(\mathrm{miss})}(\bm\beta,\bm\xi).
 \label{eq:fulllik}
\end{equation}

Condition~\eqref{eq:MAR} is missing at random in the sense of \citet{Rubin1976}, because the missingness probability depends only on the fully observed feature $Y$. 
However, the data model and missingness parameters are not distinct, since $q$ depends on the discriminant parameter $\bm\beta=\bm\beta(\bm\theta)$. Consequently, the missing label mechanism cannot be ignored for fully efficient likelihood inference. 
The likelihood $L_{\rm PC}^{(\mathrm{ig})}$ omits the information about $\bm\beta$ contained in the realised missing label pattern.

\subsection{Missing label mechanism}\label{sec:mechanism}
Throughout this subsection, $\bm\beta=\bm\beta(\bm\theta)$.
For simplicity, we use the reduced notation $\tau_i(y;\bm\beta)$ and $q(y;\bm\beta,\bm\xi)$ whenever no ambiguity can arise.
The missing label mechanism is intended to assign larger missingness probabilities to observations that are difficult to classify.
Motivated by the entropy based missing label framework of
\citet{AhfockMcLachlan2020}, we use the Shannon entropy of the posterior class membership probabilities,
\begin{equation}\label{eq:entropy}
 \en(y;\bm\beta)=-\sum_{i=1}^2\tau_i(y;\bm\beta)\log\tau_i(y;\bm\beta).
\end{equation}
Since the posterior probabilities depend on $y$ only through $d(y;\bm\beta)$, we may also write
$$ \en(d) =-\tau(d)\log\tau(d)-\{1-\tau(d)\}\log\{1-\tau(d)\}, \quad \tau(d)=\operatorname{expit}(d).$$

\begin{lemma}[Entropy and its local expansion]\label{lem:entropy}
For finite $d$,
\begin{equation}\label{eq:entropyderivative}
 \en(-d)=\en(d),\quad \en'(d)=-d\,\tau(d)\{1-\tau(d)\}.
\end{equation}
Consequently, $\en(d)$ has the unique maximum $\en(0)=\log 2$, decreases strictly with $|d|$, and tends to zero as $|d|\to\infty$. Moreover, as $d\to0$,
\begin{align}
 \en(d)&=\log 2-\frac{d^2}{8}+O(d^4),\label{eq:entropyexpand}
\end{align}
\end{lemma}

\noindent The proof is given in Appendix~\ref{app:proof-entropy}.

The probability of a missing class label to be
\begin{equation}\label{eq:qentropy}
 q(y;\bm\beta,\bm\xi) =\expit\!\left[\xi_0+\xi_1\log \en\{ d(y;\bm\beta) \}\right], \quad \bm\xi=(\xi_0,\xi_1)^{\T},\quad \xi_1\ge0.
\end{equation}
Our favourable missingness analysis focuses on true values satisfying $\xi_1\geq0$, with $\xi_1=0$ corresponding to a constant missingness probability.

For $\xi_1>0$,
$$ \frac{\partial q}{\partial \en} =\frac{\xi_1}{\en}\,q(1-q)>0,$$
so a larger posterior entropy corresponds to a larger probability of a missing class label. If $r_0>1$, then $d(t_0)=0$ and $q(y;\bm\beta,\bm\xi)$ is uniquely maximised at the Bayes decision boundary, where
$$q(t_0;\bm\beta,\bm\xi) =\expit\{\xi_0+\xi_1\log(\log2)\}. $$
If $r_0\leq 1$, no interior decision boundary exists. 
In this case, when $\xi_1>0$, both $\en\{d(y;\bm\beta)\}$ and $q(y;\bm\beta,\bm\xi)$ decrease with $y$, and their suprema are approached as $y\downarrow0$.

The coefficient $\xi_1$ controls the strength and direction of the association between classification uncertainty and label missingness. In the favourable missingness regime, 
$$0<\en(d)\leq\log 2<1. $$
Consequently, increasing $\xi_1\geq0$ while holding $\xi_0$ fixed decreases $q(y;\bm\beta,\bm\xi)$ at every $y$, with a larger decrease away from the decision boundary. 
Thus, a fixed intercept comparison changes both the expected proportion and the location of the missing labels.

Substitution of \eqref{eq:entropyexpand} into the logit of \eqref{eq:qentropy} gives
\begin{equation}\label{eq:etalocal}
 \log\frac{q(y)}{1-q(y)}=\xi_0+\xi_1\log(\log2) 
 -\frac{\xi_1}{8\log2}d(y)^2+O\{d(y)^4\}.
\end{equation}
Hence the squared discriminant model
\begin{equation}\label{eq:qd2}
 q(y;\bm\beta,\widetilde{\bm\xi}) =\expit\left\{\widetilde\xi_0 +\widetilde\xi_1d(y;\bm\beta)^2\right\}, \quad \widetilde\xi_1<0,
\end{equation}
is the local second order representation of \eqref{eq:qentropy} near $d=0$, where
$$ \widetilde\xi_0=\xi_0+\xi_1\log(\log2), \quad
 \widetilde\xi_1=-\frac{\xi_1}{8\log2}.$$
This gives the local connection with the squared discriminant missing label model used by \citet{AhfockMcLachlan2020}.

\subsection{Parameterisation for inference on the discriminant coefficients} \label{subsec:reparameterisation}
The mixture parameter $\bm\theta$ has three components, whereas the Bayes' rule depends on the two discriminant coefficients in $\bm\beta$. For calculation of the Fisher information about $\bm\beta$, we reparameterize the exponential mixture by
\begin{equation}
 \nu=\log\mu_1,\quad \bm\kappa=(\nu,\bm\beta^{\T})^{\T}.
 \label{eq:reparametrisation}
\end{equation}
The inverse transformation is
\begin{equation}
 \mu_1=e^\nu, \quad \mu_2=\{e^{-\nu}+\beta_1\}^{-1},
 \quad \log\frac{\pi_1}{\pi_2}=\beta_0-\log\frac{\mu_2}{\mu_1}.
 \label{eq:inverse-reparametrisation}
\end{equation}
Under the identifying ordering $\mu_1<\mu_2$, the parameter space is
\begin{equation}
	\mathcal K =
	\left\{ (\nu,\beta_0,\beta_1): \nu\in\mathbb R,\  \beta_0\in\mathbb R,\ 
	-e^{-\nu}<\beta_1<0 \right\}.
	\label{eq:kappa-parameter-space}
\end{equation}
The transformation is one to one on $\mathcal K$. 
The conditional distribution of $Z$ given $Y$ depends on $\bm\beta$, whereas $\nu$ is the remaining nuisance parameter for inference on the discriminant coefficients. For subsequent inference, the full parameter vector is represented as
$$ \bm\Psi = (\bm\kappa^{\T},\bm\xi^{\T})^{\T}. $$

\subsection{Estimated Bayes' rules and their error rates}
\label{subsec:estimated-rules}
Let $\widehat{\bm\beta}_{\rm CC}$ denote the maximum likelihood
estimate of $\bm\beta$ obtained by maximising
$L_{\rm CC}(\bm\theta)$. Similarly, let
$\widehat{\bm\beta}_{\rm PC}^{(\mathrm{ig})}$ and
$\widehat{\bm\beta}_{\rm PC}^{(\mathrm{full})}$ denote the maximum
likelihood estimates obtained by maximising
$L_{\rm PC}^{(\mathrm{ig})}(\bm\theta)$ and
$L_{\rm PC}^{(\mathrm{full})}(\bm\Psi)$, respectively.

Also, let
$ \widehat R_{\rm CC},\widehat R_{\rm PC}^{(\mathrm{ig})},
\widehat R_{\rm PC}^{(\mathrm{full})} $ denote the estimated Bayes' rules obtained by substituting
$\widehat{\bm\beta}_{\rm CC}$, 
$\widehat{\bm\beta}_{\rm PC}^{(\mathrm{ig})}$, and
$\widehat{\bm\beta}_{\rm PC}^{(\mathrm{full})}$, respectively, for $\bm\beta$ in $R(y;\bm\beta)$.

For any of these estimators, denoted generically by $\widehat{\bm\beta}_s$, the conditional error rate of the
corresponding estimated Bayes' rule is
\begin{align}
	\err(\widehat{\bm\beta}_s;\bm\theta)
    =\pi_1 \Pr\{\widehat R_s(Y)=2 \mid Z=1,\widehat{\bm\beta}_s\}
     +\pi_2 \Pr\{\widehat R_s(Y)=1 \mid Z=2,\widehat{\bm\beta}_s\},
	\label{eq:conditionalerror}
\end{align}
where $s\in\{{\rm CC},{\rm PC}^{(\mathrm{ig})},
{\rm PC}^{(\mathrm{full})}\}$.
Here, $(Y,Z)$ denotes a future observation independent of the
training sample. The expected error rate is
$$ \E\{\err(\widehat{\bm\beta}_s;\bm\theta)\}, $$
where the expectation is taken with respect to the training sample
distribution.

For later use, if a fitted rule has a positive threshold $t$, its conditional error rate under the data-generating parameter $\bm\theta$ is
\begin{equation}
 \err(t;\bm\theta) = \pi_1e^{-t/\mu_1} +\pi_2\{1-e^{-t/\mu_2}\}.
 \label{eq:riskthreshold}
\end{equation}
If $\bm b=(b_0,b_1)^{\T}$ satisfies $b_0>0$ and $b_1<0$, the rule
$R(y;\bm b)$ has threshold
\begin{align*}
    t(\bm b)=-\frac{b_0}{b_1}. 
\end{align*}
and conditional error rate $\err\{t(\bm b);\bm\theta\}$. 
If $b_0\leq0$ and $b_1<0$, the rule allocates every observation to class $C_2$ and its error rate is $\pi_1$.

\section{Fisher information}\label{sec:information}

\subsection{Main result}\label{sec:main}
All Fisher information matrices in this section are defined per
observation unless stated otherwise. 
Write $ \bm x(y)=(1,y)^{\T} $ and let $\I_{\rm CC}(\bm\beta)$ denote the Fisher information for $\bm\beta$ under the completely classified experiment, after eliminating the nuisance coordinate $\nu$.

The unconditional probability of a missing label is
\begin{equation}
	\gamma(\bm\theta,\bm\xi) = \E\{q(Y;\bm\beta,\bm\xi)\}
	= \int_0^\infty q(y;\bm\beta,\bm\xi) f_Y(y;\bm\theta)\,dy. \label{eq:gamma}
\end{equation}
When the meaning is clear, we write $\gamma=\gamma(\bm\theta,\bm\xi)$. 
Provided $0<\gamma<1$, the feature density among observations of $Y$ with
missing labels is
\begin{equation}
	f_{Y\mid M=1}(y) 
    = \frac{q(y;\bm\beta,\bm\xi)f_Y(y;\bm\theta)}{\gamma}.
	\label{eq:missingfeaturedistribution}
\end{equation}

Define the conditional label information matrix among observations whose labels are missing by
\begin{equation}\label{eq:Iclr}
 \I_{\rm CC}^{({\rm clr})}(\bm\beta)
 =\E\!\left[\tau_1(Y)\tau_2(Y)\bm x(Y)\bm x(Y)^{\T}\mid M=1\right].
\end{equation}
The information lost through the unavailable labels is
\begin{equation}\label{eq:Dloss}
 \bm D(\bm\beta,\bm\xi)
 =\E\!\left[q(Y)\tau_1(Y)\tau_2(Y)\bm x(Y)\bm x(Y)^{\T}\right]
 =\gamma\I_{\rm CC}^{({\rm clr})}(\bm\beta).
\end{equation}
Let $\I_{\rm PC}^{({\rm miss})}(\bm\beta)$ denote the efficient Fisher information for $\bm\beta$ obtained from the likelihood based on the missing-label indicators, after accounting for the nuisance parameter $\bm\xi$. 

\begin{theorem}[Main result]\label{thm:decomp}
Suppose that the ordered exponential mixture and the missing label mechanism are correctly specified and identifiable. Assume that the true parameter is an interior point, that $0<\gamma<1$, that the nuisance information blocks required for the Schur complements are nonsingular, and that all required derivatives and expectations exist and are finite. Then the efficient Fisher information for $\bm\beta$ under the full partially classified likelihood is
\begin{equation}\label{eq:maindecomp}
 \I_{\rm PC}^{({\rm full})}(\bm\beta) =\I_{\rm CC}(\bm\beta)
 -\gamma\I_{\rm CC}^{({\rm clr})}(\bm\beta) +\I_{\rm PC}^{({\rm miss})}(\bm\beta).
\end{equation}
\end{theorem}

In the notation for the information matrices, only the parameter about which the information is calculated is displayed in the argument. The matrices can also depend on nuisance parameters, including the parameters of the missing label mechanism.

The proof is given in Appendix~\ref{app:proof-decomp}.  
The weighting in \eqref{eq:Dloss} is essential because informative missingness changes the observations of $Y$ with missing labels.  In particular, an unavailable label near an interior decision boundary produces a relatively large conditional label information loss because $\tau_1(Y)\tau_2(Y)$ is then large.
This statement concerns the information lost through the unavailable class label, rather than the information contributed by the missing label itself.

The information loss term in \eqref{eq:maindecomp} is weighted by the conditional distribution of the observations whose labels are missing.
It cannot generally be replaced by
$$ \gamma\E\{\tau_1(Y)\tau_2(Y)\bm x(Y)\bm x(Y)^{\T}\}, $$
because observations with missing labels do not generally form a simple random subsample of the marginal feature distribution.
A sufficient condition for a gain in every direction is
\begin{equation}
	\I_{\rm PC}^{({\rm miss})}(\bm\beta)>\bm D(\bm\beta,\bm\xi),
	\label{eq:matrix-gain-condition}
\end{equation}
which implies
$$
\I_{\rm PC}^{({\rm full})}(\bm\beta)>\I_{\rm CC}(\bm\beta).
$$
This condition is stronger than necessary for improved classification. When the exponential Bayes' rule has an interior decision boundary, it depends on $\bm\beta$ through the scalar threshold $t_0$. The exact directional comparison relevant to the classification error is developed in Section~\ref{sec:are}.

\begin{corollary}[MCAR benchmark]\label{cor:mcar}
Suppose that the labels are missing completely at random, so that
\begin{equation}
 \Pr(M=1\mid Y=y,Z=i)=\gamma, \qquad 0<\gamma<1,
 \label{eq:mcar}
\end{equation}
for every $y>0$ and $i=1,2$, where $\gamma$ is variation independent
of $\bm\beta$. 
Then
$$ \I_{\rm PC}^{(\mathrm{miss})}(\bm\beta) = \bm0,$$
and
\begin{align}
 \I_{\rm PC}^{(\mathrm{full})}(\bm\beta)
 =\I_{\rm CC}(\bm\beta)-\gamma \E_{\bm\theta}\left[\tau_1(Y)\tau_2(Y)
 \bm x(Y)\bm x(Y)^\T \right] \leq \I_{\rm CC}(\bm\beta).
 \label{eq:mcar-information}
\end{align}
Under the ordered exponential mixture model, the matrix subtracted
in \eqref{eq:mcar-information} is positive definite, and hence the
inequality is strict whenever $\gamma>0$.
\end{corollary}

Thus, under MCAR, modelling the missing labels cannot recover any of the information lost through the unavailable labels.
Indeed, the full and ignorable likelihoods contain the same information about $\bm\beta$. 
The possibility that a partially classified sample is more informative than a completely classified sample therefore requires the missing label pattern to depend on the discriminant structure.

The following subsections give explicit expressions for the three
components of the decomposition, the complete classification
information, the conditionally weighted information loss and the
information contributed by the missing labels.

\subsection{Complete data information for the exponential mixture}
Under complete classification, the maximum likelihood estimators are
$$ \widehat\pi_1=\frac{n_1}{n}, \quad
\widehat\mu_i = \frac{\sum_{j=1}^nZ_{ij}Y_j}{n_i}, \quad
n_i=\sum_{j=1}^nZ_{ij}, \quad i=1,2. $$
Because the true parameter satisfies $\mu_1<\mu_2$ strictly, these estimators coincide with the maximum likelihood estimators under the ordering constraint with probability tending to one.
Let $\lambda=\log(\pi_1/\pi_2)$ and 
$\bm\theta_0=(\lambda,\mu_1,\mu_2)^{\T}$.  
Then
\begin{equation}\label{eq:thetacov}
 \sqrt n(\widehat{\bm\theta}_0-\bm\theta_0)
 \ \xrightarrow{d}\ N_3(\bm0,\bm V), \qquad
 \bm V=\diag\left\{\frac{1}{\pi_1\pi_2},\frac{\mu_1^2}{\pi_1},\frac{\mu_2^2}{\pi_2}\right\}.
\end{equation}
The Jacobian of $\bm \beta $ in \eqref{eq:beta-definition} is
\begin{equation}\label{eq:G}
 \bm G=\frac{\partial\bm\beta}{\partial\bm\theta_0^{\T}}
 =\begin{pmatrix}
     1 & -\mu_1^{-1} & \mu_2^{-1}\\
     0 &  \mu_1^{-2} & -\mu_2^{-2}
 \end{pmatrix}.
\end{equation}
Therefore
\begin{equation}\label{eq:J}
 \sqrt n(\widehat{\bm\beta}_{\rm CC}-\bm\beta)
 \xrightarrow{d} N_2(\bm0,\bm J),
 \qquad \bm J=\bm G\bm V\bm G^{\T},
\end{equation}
where
\begin{equation}\label{eq:Jexplicit}
 \bm J= \begin{pmatrix}
 \dfrac{2}{\pi_1\pi_2}
 &-\left(\dfrac{1}{\pi_1\mu_1}+\dfrac{1}{\pi_2\mu_2}\right)\\
 -\left(\dfrac{1}{\pi_1\mu_1}+\dfrac{1}{\pi_2\mu_2}\right)
 &\dfrac{1}{\pi_1\mu_1^2}+\dfrac{1}{\pi_2\mu_2^2}
 \end{pmatrix}.
\end{equation}
Hence
\begin{equation}\label{eq:IccJ}
 \I_{\rm CC}(\bm\beta)=\bm J^{-1}.
\end{equation}
This provides the complete classification benchmark used in the information comparison below.

\subsection{Conditionally weighted lost label information}
Using \eqref{eq:missingfeaturedistribution}, write
\begin{equation}\label{eq:dk}
 d_k=\gamma^{-1}\int_0^\infty y^k\tau_1(y)\tau_2(y)q(y)f_Y(y)\,dy,
 \qquad k=0,1,2.
\end{equation}
Then
\begin{equation}\label{eq:Iclrexplicit}
 \I_{\rm CC}^{({\rm clr})}(\bm\beta)
 =\begin{pmatrix}d_0&d_1\\d_1&d_2\end{pmatrix},
 \qquad
 \bm D=\gamma\begin{pmatrix}d_0&d_1\\d_1&d_2\end{pmatrix}.
\end{equation}
These integrals are finite because the exponential mixture has moments of every finite order and
$$ 0\leq q(y)\tau_1(y)\tau_2(y)\leq\frac14. $$

\subsection{Information in the missing labels}
For a logistic missing label mechanism with linear predictor
$$ q(y;\bm\beta,\bm\xi) = \expit \{\eta(y;\bm\beta,\bm\xi)\},$$
define
$$\bm a(y)=\frac{\partial\eta(y;\bm\beta,\bm\xi)}{\partial\bm \beta}, \quad 
\bm w(y)=\frac{\partial\eta(y;\bm\beta,\bm\xi)}{\partial\bm\xi}.$$
Let the information contributed by the Bernoulli likelihood of $M$
be partitioned as
$$ \bm B_{\rm miss}= \begin{pmatrix}
                        B_{\beta\beta} & B_{\beta\xi}\\
                        B_{\xi\beta}   & B_{\xi\xi}
                      \end{pmatrix}. $$
where 
\begin{align}
\bm B_{\beta\beta}&=\E[q(Y)\{1-q(Y)\}\bm a(Y)\bm a(Y)^\T], \label{eq:Bbb} \\
\bm B_{\beta\xi}&=\E[q(Y)\{1-q(Y)\}\bm a(Y)\bm w(Y)^\T],\label{eq:Bbx}\\
\bm B_{\xi\xi}&=\E[q(Y)\{1-q(Y)\}\bm w(Y)\bm w(Y)^\T].\label{eq:Bxx}
\end{align}
When $\bm\xi$ is unknown and $\bm B_{\xi\xi}$ is nonsingular, the
efficient Fisher information for $\bm\beta$ contributed by the
missing labels is
\begin{equation}\label{eq:Imiss}
 \I_{\rm PC}^{({\rm miss})}(\bm\beta)=\bm B_{\beta\beta}
 -\bm B_{\beta\xi}\bm B_{\xi\xi}^{-1}\bm B_{\xi\beta}\geq0.
\end{equation}

For the entropy model \eqref{eq:qentropy},
\begin{equation}\label{eq:logenprime}
 \frac{\partial}{\partial d}\log \en(d)=-\frac{d\,\tau_1(d)\tau_2(d)}{\en(d)}.
\end{equation}
It follows that
\begin{equation}\label{eq:awentropy}
 \bm a(y) =-\xi_1\frac{d(y)\tau_1(y)\tau_2(y)}{\en\{d(y)\}}\bm x(y),
 \quad \bm w(y)=\{1,\log \en(d(y))\}^{\T}.
\end{equation}
Substitution into \eqref{eq:Bbb}--\eqref{eq:Bxx} gives
\begin{align}
 \bm B_{\beta\beta}
 &=\xi_1^2\int_0^\infty q(y)\{1-q(y)\} \left\{\frac{d(y)\tau_1(y)\tau_2(y)}{\en\{d(y)\}}\right\}^{\!2}
 \bm x(y)\bm x(y)^{\T}f_Y(y)\,dy,\label{eq:Bbbentropy}\\
 \bm B_{\beta\xi}
 &=-\xi_1\int_0^\infty q(y)\{1-q(y)\} 
 \frac{d(y)\tau_1(y)\tau_2(y)}{\en\{d(y)\}}\bm x(y)
 \{1,\log \en(d(y))\}f_Y(y)\,dy,\label{eq:Bbxentropy}\\
 \bm B_{\xi\xi}
 &=\int_0^\infty q(y)\{1-q(y)\} \begin{pmatrix}
 1&\log \en(d(y))\\
 \log \en(d(y))&\{\log \en(d(y))\}^2
 \end{pmatrix}f_Y(y)\,dy.\label{eq:Bxxentropy}
\end{align}
For the squared discriminant model \eqref{eq:qd2},
the corresponding derivatives are
\begin{equation}\label{eq:awd2}
 \bm a(y)=2\widetilde\xi_1d(y)\bm x(y),\qquad
 \bm w(y)=\{1,d(y)^2\}^{\T}.
\end{equation}
The associated information matrices follow from \eqref{eq:Bbb}--\eqref{eq:Imiss} after replacing $q$, $\bm a$, $\bm w$, and $\bm\xi$ by their quadratic model counterparts.

\subsection{Asymptotic covariance of the full likelihood estimator}
\label{subsec:full-asymptotics}

The information decomposition also determines the asymptotic covariance of the full likelihood estimator. Let
$$ \bm\Psi_0 = (\nu_0,\bm\beta_0^\top,\bm\xi_0^\T)^\T$$
denote the true parameter value.

\begin{corollary}[Asymptotic distribution of the full likelihood
estimator] \label{cor:full-asymptotics}
Under the assumptions of Theorem~\ref{thm:decomp} and the standard regularity conditions for maximum likelihood estimation, suppose that $\bm\Psi_0$ is an interior point and that the full Fisher information matrix is nonsingular. 
Then a consistent local maximiser of \eqref{eq:fulllik} satisfies
\begin{equation}
 \sqrt n \left( \widehat{\bm\Psi}-\bm\Psi_0 \right)
 \xrightarrow{d} \mathcal N_5
 \left[ \bm0, \left\{ \I_{\rm PC}^{(\mathrm{full})}(\bm\Psi_0) \right\}^{-1}
 \right]. \label{eq:fullasymptotics}
\end{equation}
Consequently,
\begin{equation}
 \sqrt n \left( \widehat{\bm\beta}_{\rm PC}^{(\mathrm{full})} - \bm\beta_0 \right) \xrightarrow{d}
 \mathcal N_2 \left[ \bm0, \left\{ \I_{\rm PC}^{(\mathrm{full})}(\bm\beta_0) \right\}^{-1} \right],
 \label{eq:beta-full-asymptotics}
\end{equation}
where
$\I_{\rm PC}^{(\mathrm{full})}(\bm\beta_0)$ is the efficient
information obtained after eliminating the nuisance parameters
$(\nu,\bm\xi)$.
\end{corollary}
The first result follows from the standard asymptotic normality of the maximum likelihood estimator. The second follows by taking the $\bm\beta$ block of the inverse joint information matrix, which is the inverse of the efficient information obtained after eliminating $(\nu,\bm\xi)$.

\section{Asymptotic relative efficiencies of estimated Bayes' rules} \label{sec:are}
\subsection{Expected excess error and asymptotic relative efficiency} \label{subsec:expected-error}

The relative value of the partially classified sample to the completely classified sample can be measured by comparing the expected excess error rates of the estimated Bayes' rules. 
In terms of the notation introduced in Section~\ref{subsec:estimated-rules}, the relative efficiency of $\widehat R_{\rm PC}^{(\mathrm{full})}$ compared with $\widehat R_{\rm CC}$ is
\begin{equation}\label{eq:REdef}
 \RE\!\left(\widehat R_{\rm PC}^{(\mathrm{full})}\right)
 =\frac{\E\{\err(\widehat{\bm\beta}_{\rm CC};\bm\theta)\}
 -\err(\bm\theta)}{ \E\{\err(\widehat{\bm\beta}_{\rm PC}^{(\mathrm{full})};\bm\theta)\} -\err(\bm\theta)}.
\end{equation}
The corresponding asymptotic relative efficiency (ARE) is
\begin{equation}\label{eq:AREdef}
 \ARE\!\left(\widehat R_{\rm PC}^{(\mathrm{full})}\right)
 =\frac{\operatorname{AE}\{\err(\widehat{\bm\beta}_{\rm CC};\bm\theta)\}
 -\err(\bm\theta)}{\operatorname{AE}\{\err(\widehat{\bm\beta}_{\rm PC}^{(\mathrm{full})};\bm\theta)\} -\err(\bm\theta)},
\end{equation}
where $\operatorname{AE}\{\err(\widehat{\bm\beta};\bm\theta)\}$ denotes the expansion of the expected error rate through the first order term in $n^{-1}$.

We now derive the leading term in the expected excess error. 
Suppose that $A>1$, so that the true Bayes' rule has the interior decision boundary $t_0$ given in \eqref{eq:t0}.
For a fixed data generating parameter $\bm\theta$, according to \eqref{eq:riskthreshold}, 
differentiation gives
$$ \err'(t) = -\pi_1f_1(t;\mu_1) + \pi_2f_2(t;\mu_2). $$
At the Bayes boundary,
$$ \pi_1f_1(t_0;\mu_1) = \pi_2f_2(t_0;\mu_2),$$
and hence $\err'(t_0)=0$. Moreover,
\begin{equation}
 \err''(t_0) = c_0 = \pi_1f_1(t_0;\mu_1)
 \left( \frac{1}{\mu_1}-\frac{1}{\mu_2} \right)>0. \label{eq:Rsecond}
\end{equation}
The gradient of the decision boundary with respect to the discriminant coefficients is
\begin{equation}\label{eq:gthreshold}
 \bm g_t =\frac{\partial t(\bm\beta)}{\partial\bm\beta}
 =\begin{pmatrix}-1/\beta_1\\ \beta_0/\beta_1^2\end{pmatrix}
 =-\frac1{\beta_1} \begin{pmatrix}1\\t_0\end{pmatrix}.
\end{equation}

\begin{theorem}[Expected excess error of an estimated Bayes rule] \label{thm:risk}
Suppose that $\mu_1<\mu_2$ and $A>1$. Let $\widehat{\bm\beta}$ satisfy
\begin{equation}
 \sqrt n \left( \widehat{\bm\beta}-\bm\beta \right)
 \xrightarrow{\mathcal{D}}\mathcal N_2\{\bm0,\I(\bm\beta)^{-1}\}. \label{eq:asymnormal}
\end{equation}
and the required second moments are uniformly integrable. 
Then
\begin{equation}\label{eq:expectedexcess}
 \E\{\err(\widehat{\bm\beta};\bm\theta)\} -\err(\bm\theta)
 =\frac{c_0}{2n} \bm g_t^{\T}\I(\bm\beta)^{-1}\bm g_t +o(n^{-1}).
\end{equation}
\end{theorem}

\noindent The proof is given in Appendix~\ref{app:proof-error-rate}.

It follows from Theorem~\ref{thm:risk} that
\begin{equation}\label{eq:ARE}
 \ARE\!\left(\widehat R_{\rm PC}^{(\mathrm{full})}\right)
 = \frac{ \bm g_t^{\T}\I_{\rm CC}(\bm\beta)^{-1}\bm g_t}
 {\bm g_t^{\T}\I_{\rm PC}^{(\mathrm{full})}(\bm\beta)^{-1}\bm g_t}.
\end{equation}
Thus, an ARE greater than one means that the full likelihood rule has a smaller first order asymptotic expected excess error than the rule based on the completely classified sample. Condition
\eqref{eq:matrix-gain-condition} is sufficient, whereas the exact
directional condition is
\begin{equation}\label{eq:directional}
 \bm g_t^{\T} \I_{\rm PC}^{(\mathrm{full})}(\bm\beta)^{-1}\bm g_t
 < \bm g_t^{\T} \I_{\rm CC}(\bm\beta)^{-1}\bm g_t.
\end{equation}

\subsection{Sensitivity of the missing label mechanism} \label{sec:nonmonotone}

Consider the entropy model \eqref{eq:qentropy} with fixed $\xi_0$ and $\xi_1\geq0$. Since $\log\en(d)<0$, increasing $\xi_1$ decreases the missingness probability at every feature value. However, the decrease is more pronounced in regions of low posterior entropy, so the missing labels become relatively more concentrated in regions of high classification uncertainty.
For notational convenience, let
$$ \ARE(\xi_1)=\ARE\!\left(\widehat R_{\rm PC}^{(\mathrm{full})};\xi_1 \right), $$
and write $q_0=\expit(\xi_0)$. 
Then we have the following theorem.
\begin{theorem}[Fixed intercept efficiency behaviour] \label{thm:optimum}
	Fix a finite $\xi_0$. Suppose that 
    $\I_{\rm PC}^{(\mathrm{full})}(\bm\beta;\xi_1)$ is positive definite and continuous for every finite $\xi_1\geq0$.
	Then,
	\begin{enumerate}[label=(\roman*)]
		\item[(i)] At $\xi_1=0$, the missingness probability is the constant $q_0$, and the missing labels contain no information about $\bm\beta$. 
        Under the ordered exponential mixture, $\ARE(0)<1$ because the conditional label information matrix is positive definite.
		
		\item[(ii)] As $\xi_1\to\infty$,
		$$ \I_{\rm PC}^{(\mathrm{full})}(\bm\beta;\xi_1) \longrightarrow \I_{\rm CC}(\bm\beta), \quad\text{and}\quad  \ARE(\xi_1)\longrightarrow1. $$
		
		\item[(iii)] If $\ARE(\widetilde\xi_1)>1$ for some finite $\widetilde\xi_1$, then $\ARE(\xi_1)$ is non-monotone and attains a maximum greater than one at some finite $\xi_1^*>0$.
	\end{enumerate}
\end{theorem}

\noindent The proof is given in Appendix~\ref{app:proof-optimum}. Under the restriction $\xi_1\geq0$, the value $\xi_1=0$ is interpreted as the MCAR information benchmark, the usual interior point likelihood asymptotics apply when $\xi_1>0$.

A complementary comparison fixes the expected proportion of missing labels at $\gamma_0\in(0,1)$. 
For each finite $\xi_1$, the expectation in
\begin{equation}\label{eq:fixedgamma}
\E\left[\expit\{\xi_0(\xi_1;\gamma_0)+\xi_1\log\en(d(Y))\} \right]
=\gamma_0
\end{equation}
is continuous and strictly increasing in $\xi_0$. Hence there is a unique calibrated intercept $\xi_0(\xi_1;\gamma_0)$. 
This calibration is used only to select population parameter settings and it is not imposed as a constraint when calculating the information matrix. The large $\xi_1$ conclusion in Theorem~\ref{thm:optimum} does not apply because $\xi_0(\xi_1;\gamma_0)$ changes with $\xi_1$.

\subsection{Numerical evaluation of the asymptotic relative efficiency}\label{sec4.3}

We numerically evaluate the population information matrices and the resulting asymptotic relative efficiency (ARE) for the reference mixture
$$ \pi_1=\pi_2=\frac12,\qquad \mu_1=1,\qquad \mu_2=2,$$
for which
$$\bm\beta=(\log 2,-1/2)^{\T},\quad t_0=2\log 2,\quad
\err(\bm\theta)=0.375. $$
Expectations under each exponential component are computed using a
200-node Gauss-Laguerre rule. 
Increasing the order to 300 changes all reported ARE values by less than $10^{-4}$. 
We use the entropy missingness model \eqref{eq:qentropy} and eliminate the missingness parameters as in \eqref{eq:Imiss}. 
The intercept calibration is used only to select parameter settings having the prescribed missing label proportion.

To express the comparison on the scale governing the asymptotic excess error, define
$$V_s=\bm g_t^{\T}\I_s(\bm\beta)^{-1}\bm g_t,
\quad s\in\{{\rm CC},{\rm IG},{\rm Full}\}.$$
At $\gamma=0.30$, Figure~\ref{fig:information-tradeoff} decomposes the comparison into the label loss $V_{\rm IG}-V_{\rm CC}$, the gain from the missing-label indicators $V_{\rm IG}-V_{\rm Full}$, and the net gain $V_{\rm CC}-V_{\rm Full}$.  The net gain becomes positive at about $\xi_1=3.8$, where the ARE of the Full rule crosses one.

\begin{figure}[!ht]
\centering
\includegraphics[width=0.9\linewidth]{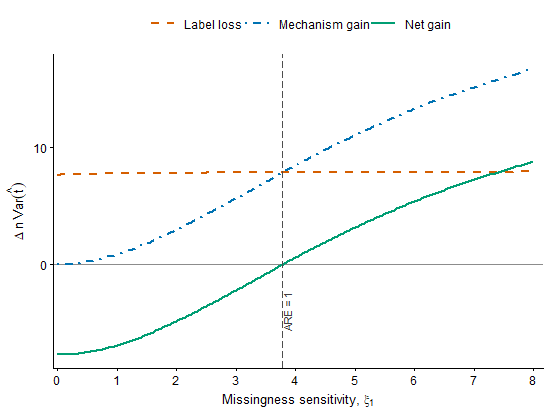}
\caption{Information trade off for the reference exponential mixture at fixed missing label proportion $\gamma=0.30$.  
The curves show the label loss $V_{\rm IG}-V_{\rm CC}$, mechanism gain $V_{\rm IG}-V_{\rm Full}$, and net gain $V_{\rm CC}-V_{\rm Full}$, where $V_s=\bm g_t^{\T}I_s^{-1}\bm g_t$.  
The vertical dashed line marks $\ARE\!\left(\widehat R_{\rm PC}^{(\mathrm{full})}\right)=1$.} \label{fig:information-tradeoff}
\end{figure}

Figure~\ref{fig:are-region} extends this calculation over the $(\xi_1,\gamma)$ plane.  For small $\xi_1$, the information in the missing-label indicators is insufficient to compensate for the unavailable class labels, so ARE is below one.  
The crossing value of $\xi_1$ increases with $\gamma$: a larger missing label proportion therefore requires a stronger dependence of missingness on classification uncertainty before the Full rule becomes more efficient than CC.  At $\gamma=0.30$, the crossing is again approximately $\xi_1=3.8$.

\begin{figure}[t]
\centering
\includegraphics[width=0.90\linewidth]{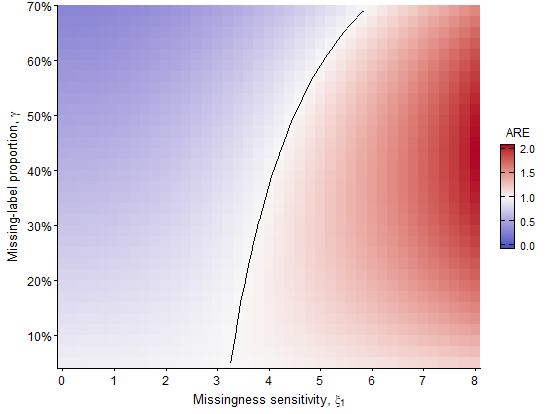}
\caption{Asymptotic relative efficiency of the Full rule over the missingness sensitivity parameter $\xi_1$ and expected missing label proportion $\gamma$.  
For each $(\xi_1,\gamma)$, the intercept $\xi_0(\xi_1;\gamma)$ is chosen to satisfy \eqref{eq:fixedgamma}.  
The black contour marks ARE$=1$; to its right, the Full rule has the smaller first order asymptotic expected excess error.}
\label{fig:are-region}
\end{figure}

These are first order comparisons of the population. Section~\ref{sec:simulations} examines their finite sample counterparts.

\section{Numerical experiments}\label{sec:simulations}

We assess the finite sample behaviour of the reference model used in Section~\ref{sec4.3}.  
Training sample sizes are $$n\in\{500,1000,1500\},$$
and each parameter setting is replicated independently $B=500$ times.  
Labels are removed according to
\begin{equation}\label{eq:simulation-missingness}
\Pr(M_j=1\mid Y_j=y)
=\expit\!\left[\xi_0+\xi_1\log\en\{d(y;\bm\beta)\}\right],
\end{equation}
where
$$\xi_1\in\{0,2,4,6,8\},\quad \gamma\in\{0.10,0.30,0.50\}.$$
For each $(\gamma,\xi_1)$, $\xi_0$ is determined by the fixed $\gamma$ calibration in \eqref{eq:fixedgamma}.  
Thus the expected missing label proportion is held fixed while $\xi_1$ controls the concentration of missing labels among observations with high posterior uncertainty. The case $\xi_1=0$ is the MCAR benchmark. The design comprises $45$ settings and $22{,}500$ simulated training samples.

\subsection{Estimation and performance measure} \label{subsec:simulation-measures}

For each simulated sample, we fit three classifiers. CC is based on the complete classification likelihood $L_{\rm CC}$. IG is based on $L_{\rm PC}^{(\mathrm{ig})}$ in \eqref{eq:iglik}, whereas Full is based on $L_{\rm PC}^{(\mathrm{full})}$ in \eqref{eq:fulllik}.   
The ordering $\mu_1<\mu_2$ is imposed during optimisation.  Likelihoods are evaluated on the log scale, multiple starting values are used for the partially classified fits, and the converged solution with the largest observed log likelihood is retained.  
The same generated complete sample is used for all three methods within each replication, so the comparisons are paired.

Classification risk is evaluated exactly under the data generating mixture.  
If $\widehat\beta_0>0$ and $\widehat\beta_1<0$, the estimated threshold is
$$ \widehat t=-\frac{\widehat\beta_0}{\widehat\beta_1},$$
with error rate
\begin{equation}\label{eq:simulation-risk}
\err(\widehat t;\bm\theta) =\pi_1\exp(-\widehat t/\mu_1)
+\pi_2\{1-\exp(-\widehat t/\mu_2)\}.
\end{equation}
If $\widehat\beta_0\le0$, the corresponding constant rule is used.

For method $s\in\{\mathrm{CC},\mathrm{IG},\mathrm{Full}\}$, let
$$ L_s^{(b)} =\err\{\widehat{\bm\beta}^{(b)}_s;\bm\theta\}
-\err(\bm\theta) $$
denote the excess classification error in replication $b$, and define $$ \overline E_s=\frac1B\sum_{b=1}^B L_s^{(b)}. $$
The empirical relative efficiencies are
\begin{equation}\label{eq:empirical-re}
\widehat{\RE}_{\mathrm{IG}:\mathrm{CC}}
=\frac{\overline E_{\rm CC}}{\overline E_{\rm IG}},
\qquad
\widehat{\RE}_{\mathrm{Full}:\mathrm{CC}}
=\frac{\overline E_{\rm CC}}{\overline E_{\rm Full}}.
\end{equation}
Values above one favour the partially classified procedure.  Monte Carlo standard errors (MCSEs) are obtained by a paired delta method.  For the Full estimator, the population benchmark shown in Figure~\ref{fig:finite-re-grid} is
$$\ARE_{\mathrm{Full}:\mathrm{CC}} =\frac{V_{\rm CC}}{V_{\rm Full}},$$
with $V_s$ defined in Section~\ref{sec4.3}.

\subsection{Results}\label{subsec:simulation-results}

Table~\ref{tab:empirical-re} reports the empirical relative efficiencies. Under the MCAR benchmark, IG and Full coincide and are less efficient than CC.  For example, at $\gamma=0.50$ the common RE is $0.519$, $0.533$, and $0.477$ for $n=500$, $1000$, and $1500$, respectively.

For informative missingness, as $\xi_1$ increases, IG remains below one throughout the reported settings, whereas Full can outperform CC. At $\gamma=0.10$ and $\xi_1=4$, the Full relative efficiencies are $1.033$, $1.028$ and $1.067$ for $n=500$, $1000$, and $1500$, respectively. At $\xi_1=8$, they increase to $1.320$, $1.383$, and $1.437$.

The transition to $\widehat{\RE}_{\mathrm{Full}:\mathrm{CC}}>1$ occurs at larger $\xi_1$ as $\gamma$ increases.  For $\gamma=0.30$, $\xi_1=4$ lies close to the transition, with Full RE $0.992$, $0.975$, and $1.065$, while at $\xi_1=6$ the corresponding values are $1.385$, $1.353$, and $1.464$.  For $\gamma=0.50$, Full remains below one at $\xi_1=4$ and takes values $0.960$, $1.407$, and $1.304$ at $\xi_1=6$.  Thus, when a larger proportion of labels is missing, a stronger dependence of missingness on classification uncertainty is required before the information carried by the missingness pattern offsets the loss of class labels.

\begin{table*}[htbp]
\centering
\caption{Empirical relative efficiencies based on $B=500$ Monte
Carlo replications, with MCSEs in parentheses. Values greater than one favour the partially classified procedure over CC. 
At $\xi_1=0$, the Full column reports the constrained MCAR benchmark and therefore coincides with IG.}
\label{tab:empirical-re}
\setlength{\tabcolsep}{4pt}
\begin{tabular}{cc cc cc cc}
\toprule
& & \multicolumn{2}{c}{$n=500$} & \multicolumn{2}{c}{$n=1000$} &
\multicolumn{2}{c}{$n=1500$}\\
\cmidrule(lr){3-4}\cmidrule(lr){5-6}\cmidrule(lr){7-8}
$\gamma$ & $\xi_1$
& $\widehat{\RE}_{\rm IG}$ & $\widehat{\RE}_{\rm Full}$
& $\widehat{\RE}_{\rm IG}$ & $\widehat{\RE}_{\rm Full}$
& $\widehat{\RE}_{\rm IG}$ & $\widehat{\RE}_{\rm Full}$\\
\midrule
0.10 & 0 & 0.884 (0.030) & 0.884 (0.030) & 0.982 (0.028) & 0.982 (0.028) & 0.883 (0.024) & 0.883 (0.024)\\
0.10 & 2 & 0.871 (0.026) & 0.939 (0.039) & 0.863 (0.026) & 0.879 (0.033) & 0.923 (0.026) & 0.993 (0.036)\\
0.10 & 4 & 0.858 (0.027) & 1.033 (0.057) & 0.914 (0.027) & 1.028 (0.049) & 0.905 (0.026) & 1.067 (0.047)\\
0.10 & 6 & 0.875 (0.025) & 1.134 (0.056) & 0.926 (0.029) & 1.099 (0.060) & 0.885 (0.024) & 1.179 (0.059)\\
0.10 & 8 & 0.845 (0.031) & 1.320 (0.082) & 0.942 (0.028) & 1.383 (0.079) & 0.929 (0.024) & 1.437 (0.084)\\
\addlinespace
0.30 & 0 & 0.716 (0.040) & 0.716 (0.040) & 0.700 (0.034) & 0.700 (0.034) & 0.761 (0.039) & 0.761 (0.039)\\
0.30 & 2 & 0.678 (0.041) & 0.739 (0.049) & 0.646 (0.039) & 0.721 (0.050) & 0.664 (0.036) & 0.715 (0.042)\\
0.30 & 4 & 0.737 (0.043) & 0.992 (0.073) & 0.679 (0.033) & 0.975 (0.061) & 0.723 (0.040) & 1.065 (0.079)\\
0.30 & 6 & 0.694 (0.031) & 1.385 (0.116) & 0.621 (0.030) & 1.353 (0.095) & 0.701 (0.036) & 1.464 (0.105)\\
0.30 & 8 & 0.628 (0.037) & 1.618 (0.145) & 0.703 (0.038) & 1.867 (0.153) & 0.669 (0.035) & 2.057 (0.162)\\
\addlinespace
0.50 & 0 & 0.519 (0.037) & 0.519 (0.037) & 0.533 (0.039) & 0.533 (0.039) & 0.477 (0.028) & 0.477 (0.028)\\
0.50 & 2 & 0.504 (0.035) & 0.609 (0.048) & 0.484 (0.032) & 0.547 (0.038) & 0.510 (0.031) & 0.601 (0.041)\\
0.50 & 4 & 0.388 (0.030) & 0.741 (0.059) & 0.487 (0.035) & 0.940 (0.082) & 0.457 (0.030) & 0.831 (0.069)\\
0.50 & 6 & 0.343 (0.031) & 0.960 (0.099) & 0.526 (0.036) & 1.407 (0.121) & 0.507 (0.037) & 1.304 (0.111)\\
0.50 & 8 & 0.461 (0.031) & 1.853 (0.172) & 0.478 (0.031) & 1.776 (0.154) & 0.410 (0.027) & 1.789 (0.171)\\
\bottomrule
\end{tabular}
\end{table*}

Figure~\ref{fig:finite-re-grid} gives the corresponding graphical comparison.  The $3\times3$ panels represent the combinations of $\gamma\in\{0.10,0.30,0.50\}$ and $n\in\{500,1000,1500\}$.  Each panel shows the empirical Full relative efficiency $\widehat{\RE}_{\mathrm{Full}:\mathrm{CC}}$, the corresponding population ARE, and the reference line RE$=1$.

The empirical curves closely follow the population trend.  For
$\gamma=0.10$, both $\widehat{\RE}_{\mathrm{Full}:\mathrm{CC}}$ and the ARE cross one at about $\xi_1=4$.  For $\gamma=0.30$, the transition remains near $\xi_1=4$, although the finite sample estimates show greater variability in this region.  
For $\gamma=0.50$, the population crossing lies between $\xi_1=4$ and $6$; the empirical crossing occurs in the same interval for $n=1000$ and $1500$, but is delayed to between $6$ and $8$ for $n=500$.
Thus, finite sample departures from the ARE are most pronounced near the RE$=1$ threshold, whereas away from this region the empirical and population curves exhibit the same increasing pattern.  At stronger dependence, the gain is substantial; for example, at $\gamma=0.30$, $\xi_1=8$, and $n=1500$, $\widehat{\RE}_{\mathrm{Full}:\mathrm{CC}}=2.057$.

\begin{figure}[htbp]
\centering
\includegraphics[width=0.98\textwidth]{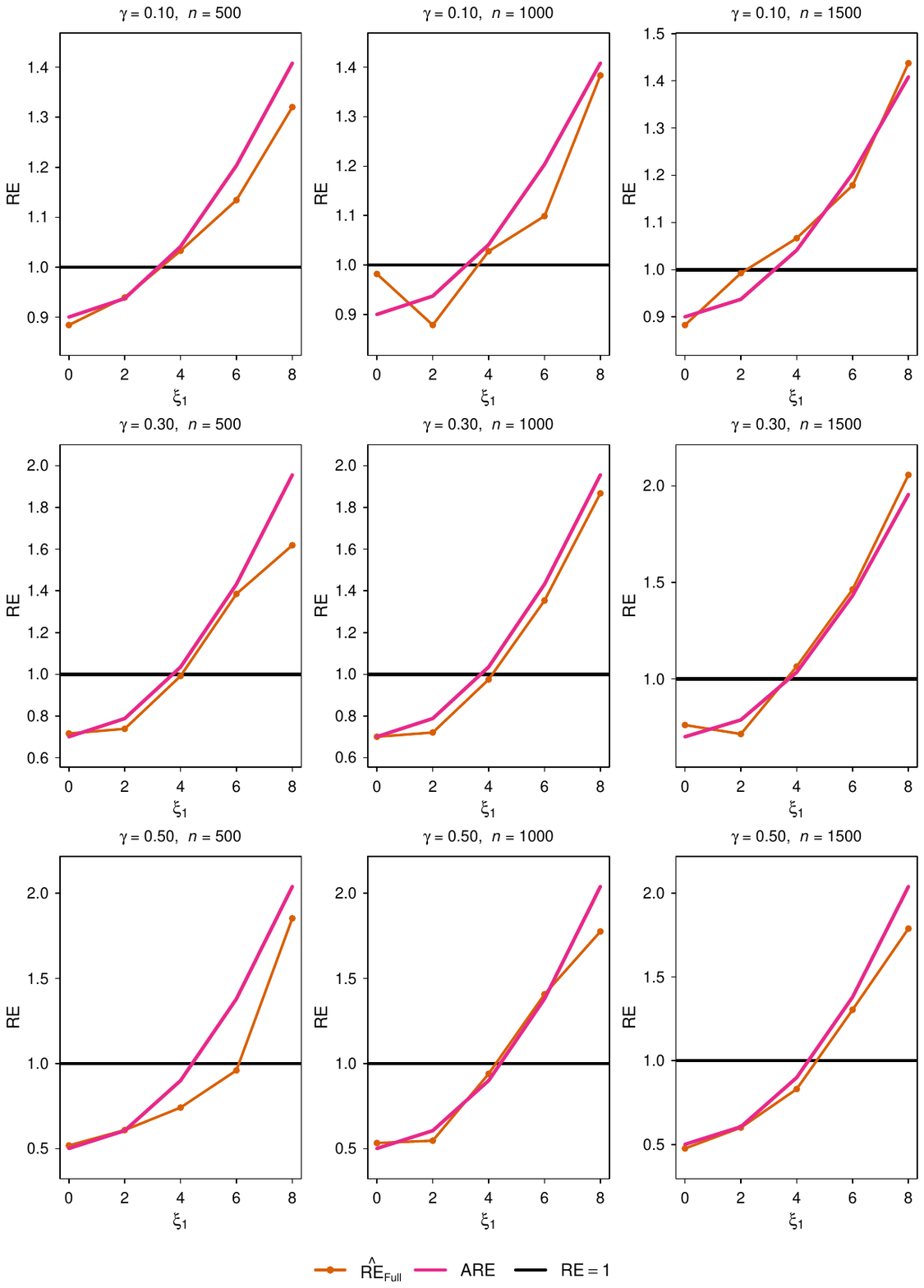}
\caption{Finite-sample relative efficiency of the Full estimator as a function of $\xi_1$. Rows correspond to $\gamma=0.10$, $0.30$, and $0.50$, and columns to $n=500$, $1000$, and $1500$. The orange curve is the empirical $\widehat{\RE}_{\mathrm{Full}:\mathrm{CC}}$ based on $B=500$ Monte Carlo replications, the magenta curve is the corresponding population ARE, and the thick black horizontal line marks RE$=1$.}
\label{fig:finite-re-grid}
\end{figure}

Therefore, the simulations support the comparison in Section~\ref{sec4.3}.  IG is less efficient than CC throughout the reported settings, whereas Full can achieve substantial gains when missingness is sufficiently concentrated near the decision boundary.  The dependence required for this gain increases with the missing-label proportion, and the largest finite-sample deviations from the ARE occur near the RE$=1$ transition.

\section{Discussion}\label{sec:discussion}

We have shown that favourable missingness can also arise in a two component exponential mixture. Unlike the homoscedastic Gaussian model, the component distributions are asymmetric and supported on the positive half line. Nevertheless, the log-posterior odds remain linear in the feature. Under the explicit interior boundary condition $r_0>1$, the Bayes' classifier is therefore determined by a single positive threshold.

The central result is the information decomposition \eqref{eq:maindecomp}. In particular, the information lost through the unavailable class labels is
$$ \bm D(\bm\beta,\bm\xi) 
= \E\!\left[ q(Y)\tau_1(Y)\tau_2(Y) \bm x(Y)\bm x(Y)^{\T} \right]
= \gamma\I_{\rm CC}^{({\rm clr})}(\bm\beta). $$
Thus, the loss is determined by the distribution of the features among observations whose labels are missing, rather than by the marginal missing label proportion alone. 
This distinction is important when missingness is concentrated in regions containing substantial conditional label information. After eliminating $\bm\xi$, the observed indicators of label missingness contribute the additional nonnegative term $\I_{\rm PC}^{({\rm miss})}(\bm\beta)$. For the exponential mixture, all terms in the decomposition can be evaluated using one dimensional integrals.

The classification error comparison follows from a second order expansion around the Bayes threshold.  Because the error rate is stationary at the optimal boundary, the leading expected excess error is of order $n^{-1}$ and depends on the asymptotic variance of the estimated threshold through \eqref{eq:expectedexcess}.  
This leads directly to the ARE in \eqref{eq:ARE}.  
The numerical calculations demonstrate that favourable regimes exist for the mixture model.  
Under a fixed intercept, any efficiency gain above one must occur at a finite value of $\xi_1$ and implies a non-monotone efficiency curve.
Under the fixed-$\gamma$ calibration considered numerically, the Full estimator becomes more efficient once the association between missingness and posterior uncertainty is sufficiently strong. 
The Monte Carlo results broadly agree with the population calculations, with the largest discrepancies occurring near the point at which the relative efficiency crosses one.

Several limitations should be noted.  
First, the model assumes uncensored exponential observations.  In survival and reliability applications, right censoring should be incorporated jointly with label missingness.  
Second, favourable performance relies on adequate specification of the missingness mechanism. Sensitivity analysis, alternative link functions and robust or bootstrap based diagnostics would therefore be useful in applications.  
Third, the comparison with the completely classified benchmark includes the additional parameter dependent observation provided by the missing label indicator $M$, and hence should not be interpreted as a violation of data processing principles.  Fourth, the present theory is one dimensional and restricted to two exponential components.  Natural extensions include gamma and Weibull mixtures, multivariate survival summaries, more than two components, covariate dependent mixing proportions, and robust sandwich inference under misspecification of the missingness mechanism.

From a design perspective, the fixed-$\gamma$ comparison in Section~\ref{sec4.3} separates the effect of where labels are missing from the expected number of missing labels, and provides a direct connection with uncertainty sampling and active learning \citep{Settles2009}.  The information decomposition suggests a local design principle: missingness should be concentrated in regions where the information contributed by the observed missing label pattern is sufficient to offset the conditionally weighted loss of class label information.  Developing globally optimal labelling policies under explicit cost, robustness, and model misspecification constraints is a natural direction for future work.

\bmhead{Acknowledgements}

This work was supported by the Australian Research Council [DP230101671] and the China Postdoctoral Science Foundation [2025M783140].

\appendix

\section{Proof of Proposition~\ref{prop:bayeserror}} \label{app:proof-bayes}
\begin{proof}
	Under the ordering $\mu_1<\mu_2$, we have $\beta_1<0$. 
    Since $\beta_0=\log r_0$, the discriminant function can be written as
	$$ d(y;\bm\beta)=\log r_0+\beta_1y, $$
	which is strictly decreasing on $(0,\infty)$.
	
	Suppose first that $r_0>1$. Then
	$$ \lim_{y\downarrow0}d(y;\bm\beta)=\log r_0>0,
	\qquad \lim_{y\to\infty}d(y;\bm\beta)=-\infty. $$
	Hence $d(y;\bm\beta)$ has the unique positive root
	$$ t_0=-\frac{\beta_0}{\beta_1}.$$
	It follows that the Bayes' rule assigns $0<y<t_0$ to class $C_1$
	and $y\geq t_0$ to class $C_2$. Therefore,
	$$ \Pr\{R(Y;\bm\beta)=2\mid Z=1\} = \Pr(Y\geq t_0\mid Z=1) =
	e^{-t_0/\mu_1}, $$
	and
	$$ \Pr\{R(Y;\bm\beta)=1\mid Z=2\} =\Pr(Y<t_0\mid Z=2)
	= 1-e^{-t_0/\mu_2}. $$
	Consequently,
	$$\err(\bm\theta)
    =\pi_1 e^{-t_0/\mu_1}+\pi_2\{1-e^{-t_0/\mu_2}\}.$$
	
	If $r_0\leq1$, then $\log r_0\leq0$. Since $\beta_1y<0$ for every $y>0$,
	$$ d(y;\bm\beta) = \log r_0+\beta_1 y<0 $$
	throughout the support. The Bayes' rule therefore assigns every observation to Class $C_2$. All observations from $C_1$ are misclassified, whereas all observations from $C_2$ are correctly classified, giving
	$$ \err(\bm\theta)=\pi_1. $$
\end{proof}

\section{Proof of Lemma~\ref{lem:entropy}}\label{app:proof-entropy}
\begin{proof}
    Since $ \tau(-d)=1-\tau(d), $ the entropy $ \en(-d)=\en(d).$
    Differentiating the entropy gives
    $$ \en'(d) = -\tau'(d)\{\log\tau(d)+1\} 
                 +\tau'(d)\big[\log\{1-\tau(d)\}+1\big], $$
    we obtain
	$$ \en'(d) = -d\,\tau(d)\{1-\tau(d)\}. $$
	Thus $\en'(d)>0$ for $d<0$, $\en'(d)<0$ for $d>0$, and
	$\en'(0)=0$. It follows that $\en(d)$ has the unique maximum
	$$ \en(0)=\log 2. $$

    As $d\to\infty$, $\tau(d)\to1$, whereas as $d\to-\infty$,
	$\tau(d)\to0$. Using the convention
	$$ \lim_{u\downarrow0}u\log u=0, $$
	we obtain
	$$ \lim_{|d|\to\infty}\en(d)=0. $$
	
	Finally, the binary logistic entropy can be written as
	$$ \en(d) = \log(1+e^d)-d\,\tau(d). $$
	As $d\to0$,
	$$ \log(1+e^d) = \log 2+\frac{d}{2}+\frac{d^2}{8} -\frac{d^4}{192}+O(d^6), $$
	and
	$$ d\,\tau(d) = \frac{d}{2}+\frac{d^2}{4} -\frac{d^4}{48}+O(d^6). $$
	Subtracting these expansions gives
	$$ \en(d) = \log 2-\frac{d^2}{8} +\frac{d^4}{64}+O(d^6), $$
	which implies \eqref{eq:entropyexpand}.
\end{proof}

\section{Proof of Theorem~\ref{thm:decomp}}\label{app:proof-decomp}
\begin{proof}
    All information matrices below are defined per observation.
    Partition the completely classified information for
    $\bm\kappa=(\nu,\bm\beta^\T)^\T$ as
    \begin{equation}\label{eq:Accblock}
        \I_{\rm CC}(\bm\kappa)= 
        \begin{pmatrix} A_{11} & A_{12}\\
                        A_{21} & A_{22} \end{pmatrix}.
    \end{equation}
    Since $\nu$ is a nuisance parameter, the efficient information for $\bm\beta$ is
    \begin{equation}\label{eq:Iccbeta}
        \I_{\rm CC}(\bm\beta) = A_{22}-A_{21}A_{11}^{-1}A_{12}.
    \end{equation}

    Because $M\perp Z\mid Y$ and the conditional scores of $Z\mid Y$ and $M\mid Y$ have mean zero given $Y$, the information from the marginal distribution of $Y$, the observed class labels, and the missing label indicators is additive, except for the $\bm\beta$--$\bm\xi$ block already included in $\bm B_{\rm miss}$.
	
    Conditional on $Y=y$, the Fisher information in an observed class label is 
    $$ \tau_1(y)\tau_2(y)\bm x(y)\bm x(y)^\T. $$
    Since the label is missing with probability $q(Y)$, the resulting information loss is
    $$\bm D=\E\left[ q(Y)\tau_1(Y)\tau_2(Y)\bm x(Y)\bm x(Y)^\T \right] =\gamma\I_{\rm CC}^{(\mathrm{clr})}(\bm\beta). $$
    Because the conditional distribution of $Z$ given $Y$ depends on $\bm\beta$ but not on $\nu$, $\bm D$ affects only the $\bm\beta\bm\beta^{\T}$ block.

    Using the information blocks in \eqref{eq:Bbb}--\eqref{eq:Imiss}, let the information contributed by the Bernoulli likelihood of $M$ be partitioned as
    $$ \bm B_{\rm miss}= \begin{pmatrix}
                          B_{\beta\beta} & B_{\beta\xi}\\
                          B_{\xi\beta}   & B_{\xi\xi}
                         \end{pmatrix}. $$
    Since the missing label model depends on $(\bm\beta,\bm\xi)$ but not on $\nu$, the joint Fisher information for 
    $\bm\Psi=(\nu,\bm\beta^{\T},\bm\xi^{\T})^{\T} $ under the full partially classified likelihood is
    $$ \I_{\rm PC}^{(\mathrm{full})}(\bm\Psi)=\begin{pmatrix}
        A_{11} & A_{12}                      &\bm0 \\
        A_{21} & A_{22}-\bm D+B_{\beta\beta} & B_{\beta\xi} \\
        \bm0   & B_{\xi\beta}                & B_{\xi\xi}
    \end{pmatrix}. $$
    Taking the Schur complement with respect to the nuisance coordinates $(\nu,\bm\xi^\T)^\T$ gives
    \begin{align*}
         \I_{\rm PC}^{(\mathrm{full})}(\bm\beta)
         &= A_{22}-A_{21}A_{11}^{-1}A_{12}-\bm D + B_{\beta\beta}
            -B_{\beta\xi} B_{\xi\xi}^{-1} B_{\xi\beta}\\
         &= \I_{\rm CC}(\bm\beta) -\bm D +\I_{\rm PC}^{(\mathrm{miss})}(\bm\beta).
    \end{align*}
    This proves Theorem~\ref{thm:decomp}.
\end{proof}

\section{Proof of Theorem~\ref{thm:risk}}
\label{app:proof-error-rate}
\begin{proof}
    The delta method gives
    $$\sqrt n(\widehat t-t_0) \xrightarrow{\mathcal{D}} 
       N\!\left(0, \bm g_t^{\T}I(\bm\beta)^{-1}\bm g_t \right).$$
    A Taylor expansion of $\err(\widehat t;\bm\theta)$ about $t_0$ gives
    $$\err(\widehat t;\bm\theta)= \err(t_0;\bm\theta) +\err'(t_0;\bm\theta)(\widehat t-t_0) +\frac12\err''(t_0;\bm\theta)(\widehat t-t_0)^2 +o_p(n^{-1}).$$

    Since $\err'(t_0;\bm\theta)=0$, taking expectations gives~ \eqref{eq:expectedexcess}.
\end{proof}

\section{Proof of Theorem~\ref{thm:optimum}}
\label{app:proof-optimum}
\begin{proof}
	When $\xi_1=0$, \eqref{eq:qentropy} reduces to
	$q(y)=\expit(\xi_0)$, which is independent of $y$ and $\bm\beta$.
	Hence
	$$ \I_{\rm PC}^{({\rm miss})}(\bm\beta)=\bm0. $$
	By Corollary~\ref{cor:mcar}, the unavailable labels remove a positive definite conditional label information matrix. Consequently,
	$$ \I_{\rm PC}^{({\rm full})}(\bm\beta) <\I_{\rm CC}(\bm\beta), $$
	and \eqref{eq:ARE} gives $\ARE(0)<1$. This proves part~(i).
    
	For part~(ii), since $\en(d)\leq\log2<1$ and $\expit(u)\leq e^u$,
	$$ q(y;\bm\beta,\bm\xi) \leq e^{\xi_0}(\log2)^{\xi_1}. $$
	Hence the conditionally weighted label loss term tends to zero as $\xi_1\to\infty$.
	Moreover,
	$$ \frac{d\,\tau_1(d)\tau_2(d)}{\en(d)} $$
	is bounded on $\mathbb R$. Hence every entry of
	$\bm B_{\beta\beta}$ is bounded in absolute value by a finite
	constant times
	$$ \xi_1^2e^{\xi_0}(\log2)^{\xi_1}\E(1+Y^2), $$
	which tends to zero. Since
	$$ \bm0\leq \I_{\rm PC}^{({\rm miss})}(\bm\beta) \leq\bm B_{\beta\beta}, $$
	we obtain
	$$ \I_{\rm PC}^{({\rm miss})}(\bm\beta)\longrightarrow\bm0.$$
	It follows from \eqref{eq:maindecomp} that
	$$ \I_{\rm PC}^{({\rm full})}(\bm\beta;\xi_1) \longrightarrow
	\I_{\rm CC}(\bm\beta), $$
	and therefore $\ARE(\xi_1)\to1$.
	
	For part~(iii), suppose that
	$\ARE(\widetilde\xi_1)>1$ for some finite
	$\widetilde\xi_1$. By part~(ii), there exists
	$M>\widetilde\xi_1$ such that
	$$ \ARE(\xi_1)<\ARE(\widetilde\xi_1), \qquad \xi_1\geq M. $$
	Continuity implies that $\ARE(\xi_1)$ attains its maximum on
	$[0,M]$. This maximum is greater than one and, by part~(i), cannot occur at $\xi_1=0$. It is therefore attained at some finite $\xi_1^*>0$. Finally,
	$$ \ARE(0)<1<\ARE(\widetilde\xi_1), \qquad \ARE(\xi_1)\longrightarrow1, $$ so $\ARE(\xi_1)$ cannot be monotone.
\end{proof}


\begin{thebibliography}{20}
\ifx \bisbn   \undefined \def \bisbn  #1{ISBN #1}\fi
\ifx \binits  \undefined \def \binits#1{#1}\fi
\ifx \bauthor  \undefined \def \bauthor#1{#1}\fi
\ifx \batitle  \undefined \def \batitle#1{#1}\fi
\ifx \bjtitle  \undefined \def \bjtitle#1{#1}\fi
\ifx \bvolume  \undefined \def \bvolume#1{\textbf{#1}}\fi
\ifx \byear  \undefined \def \byear#1{#1}\fi
\ifx \bissue  \undefined \def \bissue#1{#1}\fi
\ifx \bfpage  \undefined \def \bfpage#1{#1}\fi
\ifx \blpage  \undefined \def \blpage #1{#1}\fi
\ifx \burl  \undefined \def \burl#1{\textsf{#1}}\fi
\ifx \doiurl  \undefined \def \doiurl#1{\url{https://doi.org/#1}}\fi
\ifx \betal  \undefined \def \betal{\textit{et al.}}\fi
\ifx \binstitute  \undefined \def \binstitute#1{#1}\fi
\ifx \binstitutionaled  \undefined \def \binstitutionaled#1{#1}\fi
\ifx \bctitle  \undefined \def \bctitle#1{#1}\fi
\ifx \beditor  \undefined \def \beditor#1{#1}\fi
\ifx \bpublisher  \undefined \def \bpublisher#1{#1}\fi
\ifx \bbtitle  \undefined \def \bbtitle#1{#1}\fi
\ifx \bedition  \undefined \def \bedition#1{#1}\fi
\ifx \bseriesno  \undefined \def \bseriesno#1{#1}\fi
\ifx \blocation  \undefined \def \blocation#1{#1}\fi
\ifx \bsertitle  \undefined \def \bsertitle#1{#1}\fi
\ifx \bsnm \undefined \def \bsnm#1{#1}\fi
\ifx \bsuffix \undefined \def \bsuffix#1{#1}\fi
\ifx \bparticle \undefined \def \bparticle#1{#1}\fi
\ifx \barticle \undefined \def \barticle#1{#1}\fi
\bibcommenthead
\ifx \bconfdate \undefined \def \bconfdate #1{#1}\fi
\ifx \botherref \undefined \def \botherref #1{#1}\fi
\ifx \url \undefined \def \url#1{\textsf{#1}}\fi
\ifx \bchapter \undefined \def \bchapter#1{#1}\fi
\ifx \bbook \undefined \def \bbook#1{#1}\fi
\ifx \bcomment \undefined \def \bcomment#1{#1}\fi
\ifx \oauthor \undefined \def \oauthor#1{#1}\fi
\ifx \citeauthoryear \undefined \def \citeauthoryear#1{#1}\fi
\ifx \endbibitem  \undefined \def \endbibitem {}\fi
\ifx \bconflocation  \undefined \def \bconflocation#1{#1}\fi
\ifx \arxivurl  \undefined \def \arxivurl#1{\textsf{#1}}\fi
\csname PreBibitemsHook\endcsname

\bibitem[\protect\citeauthoryear{McLachlan}{1975}]{mclachlan1975iterative}
\begin{barticle}
\bauthor{\bsnm{McLachlan}, \binits{G.J.}}:
\batitle{Iterative reclassification procedure for constructing an asymptotically optimal rule of allocation in discriminant analysis}.
\bjtitle{Journal of the American Statistical Association}
\bvolume{70} (\bissue{350}),
\bfpage{365}--\blpage{369}
(\byear{1975})
\end{barticle}
\endbibitem

\bibitem[\protect\citeauthoryear{Ganesalingam and McLachlan}{1978}]{ganesalingam1978efficiency}
\begin{barticle}
\bauthor{\bsnm{Ganesalingam}, \binits{S.}},
\bauthor{\bsnm{McLachlan}, \binits{G.J.}}:
\batitle{The efficiency of a linear discriminant function based on unclassified initial samples}.
\bjtitle{Biometrika}
\bvolume{65} (\bissue{3}),
\bfpage{658}--\blpage{665}
(\byear{1978})
\end{barticle}
\endbibitem

\bibitem[\protect\citeauthoryear{O'Neill}{1978}]{ONeill1978}
\begin{barticle}
\bauthor{\bsnm{O'Neill}, \binits{T.J.}}:
\batitle{Normal discrimination with unclassified observations}.
\bjtitle{Journal of the American Statistical Association}
\bvolume{73} (\bissue{364}),
\bfpage{821}--\blpage{826}
(\byear{1978})
\end{barticle}
\endbibitem

\bibitem[\protect\citeauthoryear{McLachlan and Gordon}{1989}]{McLachlanGordon1989}
\begin{barticle}
\bauthor{\bsnm{McLachlan}, \binits{G.J.}},
\bauthor{\bsnm{Gordon}, \binits{R.D.}}:
\batitle{Mixture models for partially unclassified data: {A} case study of renal venous renin in hypertension}.
\bjtitle{Statistics in Medicine}
\bvolume{8} (\bissue{10}),
\bfpage{1291}--\blpage{1300}
(\byear{1989})
\end{barticle}
\endbibitem

\bibitem[\protect\citeauthoryear{Dempster et~al.}{1977}]{Dempster1977}
\begin{barticle}
\bauthor{\bsnm{Dempster}, \binits{A.P.}},
\bauthor{\bsnm{Laird}, \binits{N.M.}},
\bauthor{\bsnm{Rubin}, \binits{D.B.}}:
\batitle{Maximum likelihood from incomplete data via the {EM} algorithm}.
\bjtitle{Journal of the Royal Statistical Society: Series B (Methodological)}
\bvolume{39} (\bissue{1}),
\bfpage{1}--\blpage{22}
(\byear{1977})
\end{barticle}
\endbibitem

\bibitem[\protect\citeauthoryear{McLachlan et~al.}{2019}]{mclachlan2019finite}
\begin{barticle}
\bauthor{\bsnm{McLachlan}, \binits{G.J.}},
\bauthor{\bsnm{Lee}, \binits{S.X.}},
\bauthor{\bsnm{Rathnayake}, \binits{S.I.}}:
\batitle{Finite mixture models}.
\bjtitle{Annual Review of Statistics and its Application}
\bvolume{6} (\bissue{1}),
\bfpage{355}--\blpage{378}
(\byear{2019})
\end{barticle}
\endbibitem

\bibitem[\protect\citeauthoryear{Hady and Schwenker}{2013}]{hady2013semi}
\begin{botherref}
\oauthor{\bsnm{Hady}, \binits{M.F.A.}},
\oauthor{\bsnm{Schwenker}, \binits{F.}}:
Semi-supervised learning.
Handbook on Neural Information Processing,
215--239
(2013)
\end{botherref}
\endbibitem

\bibitem[\protect\citeauthoryear{Van et~al.}{2020}]{vanEngelenHoos2020}
\begin{barticle}
\bauthor{\bsnm{Van}, \binits{E.}},
\bauthor{\bsnm{Jesper}, \binits{E.}},
\bauthor{\bsnm{Hoos}, \binits{H.H.}}:
\batitle{A survey on semi-supervised learning}.
\bjtitle{Machine Learning}
\bvolume{109} (\bissue{2}),
\bfpage{373}--\blpage{440}
(\byear{2020})
\end{barticle}
\endbibitem

\bibitem[\protect\citeauthoryear{Ahfock and McLachlan}{2020}]{AhfockMcLachlan2020}
\begin{barticle}
\bauthor{\bsnm{Ahfock}, \binits{D.}},
\bauthor{\bsnm{McLachlan}, \binits{G.J.}}:
\batitle{An apparent paradox: {A} classifier based on a partially classified sample may have smaller expected error rate than that if the sample were completely classified}.
\bjtitle{Statistics and Computing}
\bvolume{30} (\bissue{6}),
\bfpage{1779}--\blpage{1790}
(\byear{2020})
\end{barticle}
\endbibitem

\bibitem[\protect\citeauthoryear{Louis}{1982}]{Louis1982}
\begin{barticle}
\bauthor{\bsnm{Louis}, \binits{T.A.}}:
\batitle{Finding the observed information matrix when using the {EM} algorithm}.
\bjtitle{Journal of the Royal Statistical Society Series B: Statistical Methodology}
\bvolume{44} (\bissue{2}),
\bfpage{226}--\blpage{233}
(\byear{1982})
\end{barticle}
\endbibitem

\bibitem[\protect\citeauthoryear{Castelli and Cover}{1995}]{CastelliCover1995}
\begin{barticle}
\bauthor{\bsnm{Castelli}, \binits{V.}},
\bauthor{\bsnm{Cover}, \binits{T.M.}}:
\batitle{On the exponential value of labeled samples}.
\bjtitle{Pattern Recognition Letters}
\bvolume{16} (\bissue{1}),
\bfpage{105}--\blpage{111}
(\byear{1995})
\end{barticle}
\endbibitem

\bibitem[\protect\citeauthoryear{Castelli and Cover}{1996}]{CastelliCover1996}
\begin{barticle}
\bauthor{\bsnm{Castelli}, \binits{V.}},
\bauthor{\bsnm{Cover}, \binits{T.M.}}:
\batitle{The relative value of labeled and unlabeled samples in pattern recognition with an unknown mixing parameter}.
\bjtitle{IEEE Transactions on Information Theory}
\bvolume{42} (\bissue{6}),
\bfpage{2102}--\blpage{2117}
(\byear{1996})
\end{barticle}
\endbibitem

\bibitem[\protect\citeauthoryear{Lyu}{2024}]{Lyu2024}
\begin{barticle}
\bauthor{\bsnm{Lyu}, \binits{Z.}}:
\batitle{Analysis of estimating the {B}ayes' rule for {G}aussian mixture models with a specified missing-data mechanism}.
\bjtitle{Computational Statistics}
\bvolume{39} (\bissue{7}),
\bfpage{3727}--\blpage{3751}
(\byear{2024})
\end{barticle}
\endbibitem

\bibitem[\protect\citeauthoryear{Lyu et~al.}{2024}]{LyuAhfockThompsonMcLachlan2024}
\begin{barticle}
\bauthor{\bsnm{Lyu}, \binits{Z.}},
\bauthor{\bsnm{Ahfock}, \binits{D.}},
\bauthor{\bsnm{Thompson}, \binits{R.}},
\bauthor{\bsnm{McLachlan}, \binits{G.J.}}:
\batitle{Semi-supervised {G}aussian mixture modelling with a missing-data mechanism in {R}}.
\bjtitle{Australian \& New Zealand Journal of Statistics}
\bvolume{66} (\bissue{2}),
\bfpage{146}--\blpage{162}
(\byear{2024})
\end{barticle}
\endbibitem

\bibitem[\protect\citeauthoryear{Wu et~al.}{2026}]{WuWangMcLachlan2026}
\begin{barticle}
\bauthor{\bsnm{Wu}, \binits{J.}},
\bauthor{\bsnm{Wang}, \binits{Y.G.}},
\bauthor{\bsnm{McLachlan}, \binits{G.J.}}:
\batitle{Informative missingness and its implications in semi-supervised learning}.
\bjtitle{The Innovation Informatics}
\bvolume{2} (\bissue{2}),
\bfpage{100033}
(\byear{2026})
\end{barticle}
\endbibitem

\bibitem[\protect\citeauthoryear{Jewell}{1982}]{Jewell1982}
\begin{botherref}
\oauthor{\bsnm{Jewell}, \binits{N.P.}}:
Mixtures of exponential distributions.
The Annals of Statistics,
479--484
(1982)
\end{botherref}
\endbibitem

\bibitem[\protect\citeauthoryear{Teicher}{1961}]{Teicher1961}
\begin{barticle}
\bauthor{\bsnm{Teicher}, \binits{H.}}:
\batitle{Identifiability of mixtures}.
\bjtitle{The Annals of Mathematical Statistics}
\bvolume{32} (\bissue{1}),
\bfpage{244}--\blpage{248}
(\byear{1961})
\end{barticle}
\endbibitem

\bibitem[\protect\citeauthoryear{Yakowitz and Spragins}{1968}]{YakowitzSpragins1968}
\begin{barticle}
\bauthor{\bsnm{Yakowitz}, \binits{S.J.}},
\bauthor{\bsnm{Spragins}, \binits{J.D.}}:
\batitle{On the identifiability of finite mixtures}.
\bjtitle{The Annals of Mathematical Statistics}
\bvolume{39} (\bissue{1}),
\bfpage{209}--\blpage{214}
(\byear{1968})
\end{barticle}
\endbibitem

\bibitem[\protect\citeauthoryear{Rubin}{1976}]{Rubin1976}
\begin{barticle}
\bauthor{\bsnm{Rubin}, \binits{D.B.}}:
\batitle{Inference and missing data}.
\bjtitle{Biometrika}
\bvolume{63} (\bissue{3}),
\bfpage{581}--\blpage{592}
(\byear{1976})
\end{barticle}
\endbibitem

\bibitem[\protect\citeauthoryear{Settles}{2009}]{Settles2009}
\begin{botherref}
\oauthor{\bsnm{Settles}, \binits{B.}}:
Active learning literature survey.
Technical Report
(2009)
\end{botherref}
\endbibitem

\end{thebibliography}

\end{document}